\documentclass[conference,compsoc]{IEEEtran}
\IEEEoverridecommandlockouts
\usepackage{cite}
\usepackage{amsmath,amssymb,amsfonts}
\usepackage{algorithmic}
\usepackage{graphicx}
\usepackage{textcomp}
\usepackage{xcolor}
\usepackage{amsthm,adjustbox,enumitem}
\usepackage{booktabs}
\usepackage{stmaryrd}
\def\BibTeX{{\rm B\kern-.05em{\sc i\kern-.025em b}\kern-.08em
    T\kern-.1667em\lower.7ex\hbox{E}\kern-.125emX}}
    
\usepackage{hyperref}

\definecolor{darkblue}{rgb}{0, 0, 0.5}
\hypersetup{colorlinks=true,citecolor=darkblue,linkcolor=darkblue,urlcolor=darkblue}
    
\theoremstyle{plain}
\newtheorem{thm}{Theorem}[section]

\theoremstyle{definition}

\begin{document}

\title{CoVer: Collaborative Light-Node-Only Verification and Data Availability for Blockchains
}

\author{\IEEEauthorblockN{Steven Cao, Swanand Kadhe, Kannan Ramchandran}
\IEEEauthorblockA{\textit{Department of Electrical Engineering and Computer Sciences} \\
\textit{University of California, Berkeley}\\
\{stevencao, swanand.kadhe, kannanr\}@berkeley.edu}}

\maketitle

\begin{abstract}
Validating a blockchain incurs heavy computation, communication, and storage costs. As a result, clients with limited resources, called \textit{light nodes}, cannot verify transactions independently and must trust \textit{full nodes}, making them vulnerable to security attacks. Motivated by this problem, we ask a fundamental question: can light nodes securely validate without any full nodes? We answer affirmatively by proposing CoVer, a decentralized protocol that allows a group of light nodes to collaboratively verify blocks even under a dishonest majority, achieving the same level of security for block validation as full nodes while only requiring a fraction of the work. In particular, work per node scales down proportionally with the number of participants (up to a log factor), resulting in computation, communication, and storage requirements that are \textit{sublinear in block size}. Our main contributions are light-node-only protocols for fraud proofs and data availability.
\end{abstract}

\section{Introduction}

\subsection{Motivation}
\label{sec:motivation}

Blockchain participants that independently validate and store all the blocks are called \textit{full nodes} and are vital to the security of the network. However, as adoption of cryptocurrencies grows, the burden of running a full node will become infeasible for most users due to higher throughput, causing more participants to operate as \textit{light nodes}. Light nodes or light clients, based on the simplified payment verification (SPV) protocol proposed by Nakamoto~\cite{bitcoin}, store only block headers and do not validate transactions. As a result, they cannot contribute to the security of the network and are vulnerable to several security and privacy attacks (see, e.g.,~\cite[Chapter 6]{Karame:16}). In particular, light nodes accept the longest header chain, assuming that any \textit{invalid} chain, or a chain containing fraudulent transactions, will not be mined upon. Therefore, in the status quo, their security depends on the longest chain being valid, requiring the strong assumption that the majority of miners are honest. 

In \cite{fraud}, Al-Bassam et al.\ propose protocols for \textit{fraud proofs} and \textit{data availability}, allowing light nodes to reject headers of invalid blocks by receiving compact proofs of their invalidity from full nodes. While block proposal protocols still require an honest majority (e.g.\ to prevent forking-based double spending attacks), they solve the verification problem with an honest minority of full nodes. However, the protocols still require each light node to be in communication with an honest full node, increasing the burden on full nodes and limiting the protocol's scalability.

Motivated by these challenges, we ask a fundamental question: {\it can light nodes securely validate without any full nodes?} Formally, we characterize a light node as one that bears computation, communication, and storage costs \emph{sublinear in block size}. We answer affirmatively by proposing a decentralized and permissionless light-node-only {\bf Co}llaborative {\bf Ver}ification protocol, called {\bf CoVer}, which enables a group of light nodes to collaboratively validate blocks even under a dishonest majority, resulting in the same level of security as full nodes.\footnote{Like \cite{fraud} and \cite{codedmerkle}, we restrict our attention to block verification and not block proposal. Forking-based double spending attacks are still possible with a dishonest majority of miners.} Furthermore, the protocol is flexible in that participants can perform as much or as little validation as they wish, with more active participants contributing more to the security of the protocol. The result is a blockchain network that accommodates and utilizes the computing resources of participants of all sizes.


\subsection{Overview}\label{section:overview}
The key ingredients in CoVer are light-node-only protocols for fraud proofs and data availability (see Fig.~\ref{fig:system-model} for an overview). At a high level, each node verifies a small part of each block and broadcasts a \textit{fraud proof}, i.e.\ a proof of block invalidity, for any invalid transaction. A node rejects the block if it receives a valid fraud proof and accepts it otherwise. Then, the protocol is secure as long as each transaction is validated by at least one honest participant. 

To realize this idea, we propose protocols to solve the following three challenges. The first challenge is that light nodes should be able to validate individual transactions without needing to process the whole block or store the entire state.
We overcome this challenge by proposing a \textit{new block structure and fraud proof protocol} (Section~\ref{section:validation}).

The second challenge is to guarantee \textit{data availability}~\cite{fraud}: if a miner includes an invalid transaction in the Merkle root of the block but does not make this transaction available for download, then no node can produce a fraud proof for it. Therefore, light nodes should accept a header only if all of its transactions are available, and they must check availability while only downloading a small amount of data. To solve this problem, the protocol in \cite{fraud} requires miners to encode the block data using an erasure code. Then, to hide any single transaction, a miner must prevent decoding by either (1) hiding at least a constant fraction of the block, in which case a light node can catch unavailability with high probability by randomly sampling a constant number of shares, or (2) constructing the coded data incorrectly, in which case full nodes can produce proofs of coding fraud. The bottleneck in this approach is decoding the block, which is at least linear in block size and therefore infeasible for light nodes. We address this challenge by proposing a \textit{secure collaborative decoding protocol} on top of the recently proposed LDPC-coded Merkle tree~\cite{codedmerkle} (Section~\ref{section:coding}).

The third challenge is that light nodes cannot download the entire block. To address this challenge, we propose and analyze a simple modification of the gossip protocol, deemed \textit{selective broadcast}, in which nodes only gossip the data that they need (Section~\ref{section:broadcast}).

Together, these protocols empower a group of light nodes to collaboratively validate a block with each node performing only a small amount of \textit{work}, i.e.\ computation, communication, and storage. Specifically, if each block has $L$ transactions and there are $N_h$ honest light nodes, each light node performs $\tilde{O}(L/N_h)$ work, where $\tilde{O}$ denotes equivalence up to log factors. Due to the connectivity required for the selective broadcast communication scheme, this savings factor is capped at $\sqrt{L}$, resulting in $\tilde{O}(\sqrt{L})$ work per node for large validation pools.\footnote{This ceiling can be bypassed by alternate communication schemes that may be possible given stronger assumptions.} These protocols are also decentralized and permissionless while requiring only an honest minority. Table~\ref{table:summary} summarizes the properties of CoVer while comparing it to other approaches.

\begin{table*}[t]
	\begin{center}
		\caption{\label{table:summary}Summary of approaches for light node security\vspace{-10pt}}
		\resizebox{\linewidth}{!}{\begin{tabular}{ l|l|l|l|l|l|l|l }
			& \shortstack[l]{Full node \\ computation} &\shortstack[l]{Full node \\ communication} &\shortstack[l]{Light client \\ computation} & \shortstack[l]{Light client \\ communication} & \shortstack[l]{Light client storage} & Block size & \shortstack[l]{Light client \\ security dependence} \\ 
			\hline
			\hline
			Status quo \cite{bitcoin} & $O(L)$ & $O(L)$ & O(1) & O(1) & $O(T)$ & $O(L)$ & \shortstack[l]{Miners \\ (honest majority)} \\ 
			\hline
			\shortstack[l]{Fraud proofs \cite{fraud} + \\ 2D Reed Solomon \cite{fraud}$^\alpha$} & \shortstack[l]{$O(L \log A)$ + \\ $O(L^{1.5})$}& $O(L)$ & \shortstack[l]{$O(\log (AL))$ + \\ $O(\sqrt{L} \log L)$ } & \shortstack[l]{$O(\log (AL))$ + \\ $O(\sqrt{L} \log L)$ } & \shortstack[l]{$O(T)$ + \\ $O(\sqrt{L}T)$} & $O(L)$ & \shortstack[l]{Full nodes \\ (honest minority)} \\ 
			\hline
			\shortstack[l]{Fraud proofs \cite{fraud} + \\ Coded Merkle Tree \cite{codedmerkle}} & \shortstack[l]{\phantom{$O(L^{1.5})$ + } \\ $O(L \log A)$} & $O(L)$ & $O(\log (AL))$ & $O(\log (AL))$ & $O(T)$ & $O(L)$ & \shortstack[l]{Full nodes \\ (honest minority)} \\ 
			\hline
			\hline
			CoVer (this work) & N/A & N/A & $O\left(\frac{1}{k} L\log L\right)$ &  $O(k \log N_h)$ &  $O\left(T + \frac{1}{k} L\log L\right)$ & $O(L \log L)$ & \shortstack[l]{Each other \\ (honest minority)} \\
			\hline
			\multicolumn{8}{l}{}\\
			\multicolumn{8}{l}{$L$: transactions per block, $N_h$: number of honest participants, $A$: total number of accounts, $T$: number of rounds, $k$: division of work parameter, capped at $\frac{\log N_h}{N_h}$} \\
			\multicolumn{8}{l}{We compare to the original fraud proofs with two versions of data availability, 2D Reed Solomon \cite{fraud} and the coded Merkle tree \cite{codedmerkle}.}\\
			\multicolumn{8}{l}{$^\alpha$For fraud proofs +  2D Reed Solomon, we split the costs into those due to validation of transactions (upper row) and data availability (lower row).}
		\end{tabular}}
	\end{center}
	\vspace{-10pt}
\end{table*}

\subsection{Problem Setup and Objectives}
\noindent\textbf{System Model:}
We consider two types of nodes:
\begin{enumerate}
    \item \textit{Miners}: nodes that produce and broadcast new blocks.
    \item \textit{Validators}: nodes that validate new blocks. We also refer to validators as \textit{participants}.
\end{enumerate}
We assume that all participants are \textit{light nodes}. Specifically, we model a light node as a node having computation, communication, and storage capabilities sublinear in the size of each block. We note, however, that nodes with higher capacity can also participate in the protocol. For example, a node with twice the capability of a normal light node can act as two light nodes and perform twice the work.

We assume bounded network delay such that if a node sends a message, it can be received by all other nodes connected to the sender within some maximum delay $\Delta$. For analysis, we model the connectivity between honest nodes as an Erd\H{o}s-R{\'e}nyi random graph, where each pair of nodes is connected independently with some probability $p$ \cite{graph}. This connection probability is a parameter of the protocol, and we analyze the required connectivity in Section~\ref{section:connectivity}. We make no assumptions on the connectivity of dishonest nodes.

\begin{figure}[!t]
    \centering
    \includegraphics[width=0.75\linewidth]{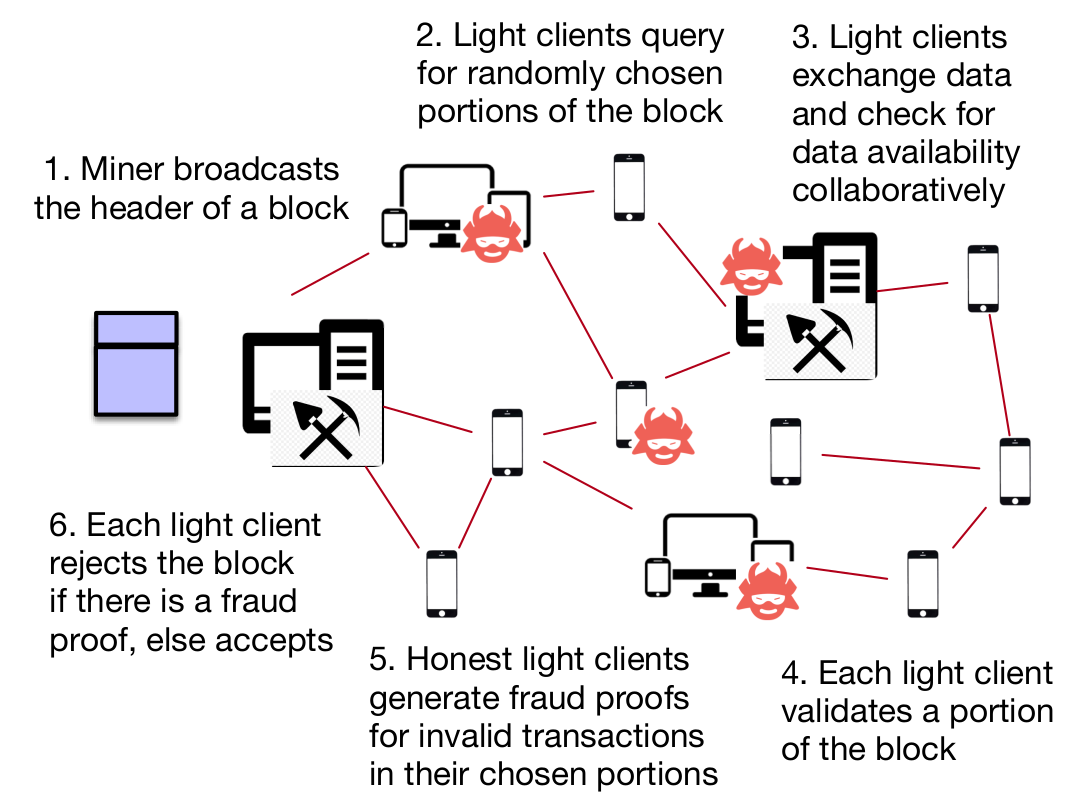}
    \caption{An overview of CoVer. We consider a network with miners, i.e.\ block producers, and light node validators, which perform computation, communication and storage sublinear in block size. The key ingredients are light-node-only protocols for data availability and fraud proofs. }
    \label{fig:system-model}
    \vspace{-10pt}
\end{figure}

\vspace{4pt}
\noindent\textbf{Threat Model:} We assume that any nodes can be dishonest and act adversarially. Dishonest nodes may deviate from the protocol in arbitrary manner and may collude with each other. Further, we assume that dishonest nodes can perform any polynomial-time computations but cannot invert hash functions. We require an honest minority, meaning that for the protocol to function, there must be honest participants doing enough work such that every transaction in a block is validated by at least one participant. The required minimum number is analyzed in Section~\ref{section:coverage}. Like \cite{fraud}, we only solve verification. Forking-based double spending attacks are still possible with a dishonest majority of miners.

\vspace{4pt}
\noindent\textbf{Objectives:} In every round, a miner produces a block and broadcasts the block along with its header. A malicious miner may only broadcast subset of its data. 
Then, each participant executes the protocol and chooses to either accept or reject the block. The protocol should be
\begin{enumerate}
    \item \textit{Correct}: all participants accept the block if and only if
    \begin{enumerate}
        \item The block is \textit{valid} (contains only valid transactions).
        \item The block is \textit{available}, meaning that the miner broadcasted enough of the block data such that all of the data can be decoded (see Section~\ref{section:coding}).
    \end{enumerate}
    \item \textit{Light-node-only}: each participant performs computation, communication, and storage sublinear in the number of transactions in the incoming block.
    \item \textit{Permissionless}: nodes can freely join and leave.
    \item \textit{Decentralized}: the requirements of the protocol should be uniform over all participants.
    \item \textit{Secure under honest minority}: for block verification, the number of honest participants required for the protocol to be correct should not depend on the number of total participants.
\end{enumerate}
We require correctness only with high probability, rather than in the worst case. We will use $\lambda$ to denote a security parameter such that correctness holds with probability at least $(1- e^{-\lambda})$. The specific high probability expressions are calculated in the correctness proofs (Section~\ref{section:correctnessproof}).

\section{Collaborative Verification}\label{section:validation}
\subsection{Proposed Block Structure}
To enable fraud proofs, Al-Bassam et al.~\cite{fraud} modify the block to include the root of a sparse Merkle tree that contains the state (UTXOs or account balances). However, this approach does not accommodate light clients because the state cannot be divided. Therefore, we propose an alternate block structure and fraud proof protocol. First, as usual, the header of the $i$-th block, denoted $\texttt{header}(i)$, contains:
\begin{itemize}
	\item $\texttt{prevHash}$: hash of the previous block.
	\item $\texttt{root}$: root of the Merkle tree containing transactions. 
	\item $\texttt{len}$: number of transactions.
	\item $\texttt{other}$: other data (like the nonce for proof-of-work).
\end{itemize}
We will denote the Merkle proof of the transaction $\texttt{txn}$ in $\text{block $i$}$ as $\texttt{proof}(\texttt{txn} \to \texttt{header}(i).\texttt{root})$. Next, as usual, each transaction $\texttt{txn}$ in block $i$ contains:
\begin{itemize}
    \item \texttt{txid}: transaction id (e.g., 32 byte transaction hash) 
	\item \texttt{sender}: the id (e.g.,\ public key) of the sender.
	\item \texttt{signature}: the sender's signature.
	\item \texttt{outputs}: a list of transaction outputs, numbered $1,2,...$, which specify the recipient and the amount.
\end{itemize}
To enable fraud proofs (see Section~\ref{section:coverfraud}), we require that each transaction reference and provide Merkle proofs for the UTXOs (unspent transaction outputs) that fund it:
\begin{itemize}
	\item \texttt{inputs}: a list of past TXOs $(\texttt{txid}, j)$, i.e.\ the $j$th output of $\texttt{txn}_\texttt{txid}$, where the sender received money. 
	\item \texttt{inputProofs}: for each $(\texttt{txid}, j)$ in $\texttt{inputs}$, a Merkle proof
	linking each input to the \texttt{root} in some past $\texttt{header}(k)$ in the chain, where $k < i$. 
\end{itemize}
Given this block structure, a transaction is valid if
\begin{enumerate}
	\item The signature is valid.
	\item The sum of inputs equals the sum of outputs.\footnote{If the inputs are greater, the sender can send herself the change. For simplicity, we do not consider transaction fees.}
	\item The input Merkle proofs are valid.
	\item The corresponding TXOs are unspent. 
\end{enumerate}
Any node can check (1)-(3) with just the header chain, but checking (4) requires some knowledge of the current \textit{state}, i.e.\ some representation of the transaction history. 
For our protocol, validators store the state as a hash table \texttt{spentTXOs} mapping the IDs of spent TXOs to the transaction that spent it, along with a Merkle proof of that transaction to enable fraud proofs (see Section~\ref{section:coverfraud}).

Assuming that each transaction has a constant number of outputs, after $T$ blocks with $L$ transactions per block, the size of this table will be $O(TL\log L)$. We will reduce the size of this table by removing the dependence on $T$ in Section~\ref{section:lessstorage}, and we will further reduce to $O((L/k)\log L)$ by dividing the work among the participants in Section~\ref{section:collabvalprotocol}.

Then, a node can validate \texttt{txn} in block $i$ as follows:
\begin{equation*}
    \texttt{is\_valid\_txn}(\texttt{txn}, \texttt{spentTXOs}) \in \{\texttt{True}, \texttt{False}\}
\end{equation*}
\begin{enumerate}[label=\texttt{\arabic*:}]
	\item Check if the \texttt{signature} is valid.
	\item Check that sum of \texttt{inputs} equals sum of \texttt{outputs}.
	\item Check that the \texttt{inputProofs} are valid.
	\item Check for double payment: check that each TXO in \texttt{inputs} is not in the table \texttt{spentTXOs}.
	\item If any check fails, return \texttt{False} and broadcast a \texttt{fraudProof}, as described in the following section.
	\item If all checks succeed, return \texttt{True}.
\end{enumerate}
If the block is valid, then for each \texttt{txn}, a node can update its state storage by including the spent TXOs as follows:
\begin{equation*}
    \texttt{update\_state}(\texttt{txn}, \texttt{spentTXOs})
\end{equation*}
\begin{enumerate}[label=\texttt{\arabic*:}]
    \item For each $(\texttt{txid}, j)$ in \texttt{inputs}, insert the entry \\
$(\texttt{txid}, j)$ : $(\texttt{txn}, \texttt{proof}(\texttt{txn} \shortrightarrow \texttt{header}(i).\texttt{root}))$ \\
into the hash table \texttt{spentTXOs}.
\end{enumerate}

\subsection{CoVer Fraud Proofs}\label{section:coverfraud}
If a transaction in the block is invalid, then a validator can produce a fraud proof, which proves that the block contains an invalid transaction and can be checked by any node storing the header chain. A \texttt{fraudProof} contains:
\begin{itemize}
	\item \texttt{invalidTxn}: the invalid transaction.
	\item \texttt{invalidTxnProof}: the Merkle proof to header $i$.
	\item \texttt{pastTxn}: if the transaction is invalid due to double spending, the proof contains a past transaction from the same \texttt{spender} with at least one colliding input.
	\item \texttt{pastTxnProof}: the Merkle proof to some past header $k$, where $k < i$.
\end{itemize}
Then, any recipient can check the fraud proof as follows:
\begin{align*}
    \texttt{is\_valid\_fraudProof}(\texttt{fraudProof}) \in \{\texttt{T}, \texttt{F}\}
\end{align*}
\begin{enumerate}[label=\texttt{\arabic*:}]
    \item Check \texttt{invalidTxnProof} and \texttt{pastTxnProof}. If either proof is invalid, return \texttt{False}.
    \item Perform the first three checks of \texttt{is\_valid\_txn}(\texttt{invalidTxn}, $\bullet$). If any check fails, then the transaction is invalid, so return \texttt{True}.
    \item Check for double spending: if the intersection \texttt{invalidTxn.inputs} $\cap$ \texttt{pastTxn.inputs} is non-empty, then return \texttt{True}.
    \item Otherwise, return \texttt{False}.
\end{enumerate}

Assuming that the number of inputs and outputs is constant and hash table lookup is constant time, the heaviest computation is checking a constant number of Merkle proofs, so the runtimes of \texttt{is\_valid\_fraudProof} and \texttt{is\_valid\_txn} are both $O(\log L)$.

\begin{figure}[!t]
    \centering
    \includegraphics[width=\linewidth]{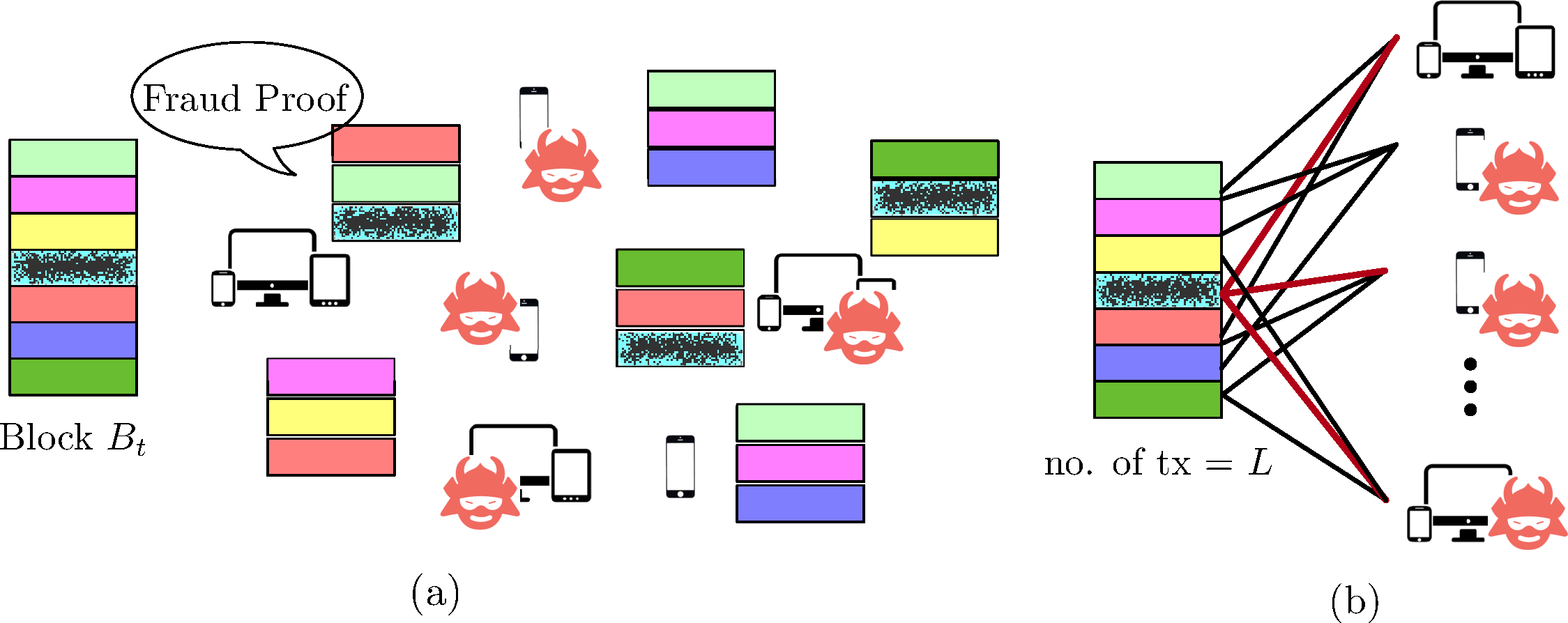}
    \caption{Collaborative verification using fraud proofs. Each light node will verify roughly $1/k$ fraction of each block on average. Honest light nodes  generate fraud proofs for any invalid transactions that they verify. Every invalid transaction will be caught with a fraud proof as long as each of the $k$ sections is \textit{covered} (i.e., validated) by at least one honest participant.}
    \label{fig:collaborative-verification}
    \vspace{-10pt}
\end{figure}

\subsection{Collaborative Validation}\label{section:collabvalprotocol}
Note that a node can verify a transaction as long as it has the history of its sender, so a node can choose a segment $[a,b]$ of account numbers and only validate transactions sent by those accounts.\footnote{To facilitate partial validation, the transactions in a block should be sorted by the sender's account number. A ``sorting fraud proof'' would contain two transactions that are out of order, plus their Merkle proofs.} We will divide the account numbers into $k$ equal-sized sections and each node will choose one section at random, resulting in the following protocol (see Fig.~\ref{fig:collaborative-verification}):
\begin{enumerate}
	\item Download the \texttt{header}.
	\item Download the transactions \texttt{txn} where \texttt{txn.sender} is in the node's chosen account section. The selective downloading scheme is described in Section~\ref{section:broadcast}.
	\item Validate the transactions and broadcast a \texttt{fraudProof} if any is invalid.
	\item Wait some fixed delay. Reject the block if you receive a valid \texttt{fraudProof}, otherwise accept it.
	\item If the block is accepted, update \texttt{spentTXOs} to include the newly spent \texttt{inputs} for each \texttt{txn}.
\end{enumerate}
Note that malicious nodes cannot generate fraud proofs for valid transactions, and a single valid fraud proof is sufficient to discard the block. Then, the scheme is secure as long as each of the $k$ sections is covered by at least one honest participant. We analyze the probability of coverage in Section~\ref{section:coverage} and show that we need $N_h > k(\log k + \lambda)$ honest nodes, each validating $1/k$th of the block, to cover the block with probability at least $1-e^{-\lambda}$. Therefore, the fraction $1/k$ of work required scales down proportionally with the number $N_h$ of honest nodes, up to a log factor. Honest nodes with more resources can verify more than one section, counting as multiple nodes in the calculation above. 

\subsection{Reducing State Storage and Switching Sections}\label{section:lessstorage}
One shortcoming of the above approach is that the amount of storage grows linearly with the length of the blockchain. Specifically, each node stores $O(L/k)$ transactions per block on average, each of which is size $O(\log L)$, resulting in $O(T \frac{L}{k} \log L)$ storage after $T$ rounds. To reduce this storage, we add an \textit{expiration time} inspired by~\cite{vault}, where each transaction can only be spent up to $\tau$ blocks after it is mined (note that the TXO owner can pay herself to avoid expiration). Then, storage is reduced to $O(\tau \frac{L}{k} \log L)$.

While each validator can choose to stick to one account section forever, it may be desirable from a security perspective for validators to periodically switch sections, e.g. in the case of slowly adaptive adversaries. Because TXOs are only valid up to $\tau$ blocks after they are mined, two transactions whose inputs intersect must be within $\tau-1$ blocks of each other. Then, a validator can join a new section as follows:
\begin{enumerate}
	\item For time steps $t, t+1, ..., t + (\tau -1)$, download the transactions for the new section.
	\item At time step $t + \tau$, the validator now has enough history to validate the new section.
\end{enumerate}

\section{Collaborative Data Availability Coding}\label{section:coding}

\subsection{Coding Structure}\label{section:codeoverview}

As described in Section~\ref{section:overview}, a fraud proof cannot be produced for a hidden transaction, so light nodes must be assured that the entire block is available \textit{without downloading the entire block}. To solve this problem, known as \textit{data availability}, the high-level idea in \cite{fraud} is to require miners to encode their block with an erasure code. This idea is improved in \cite{codedmerkle} via a coding scheme that combines Low-Density Parity-Check (LDPC) codes with Merkle trees. 

First, we review the LDPC erasure code~\cite{ldpc}~\cite{Richardson:08}. Starting with $n$ data symbols, a rate-$1/2$ LDPC code produces $n$ additional coded symbols, resulting in $2n$ symbols total. The code is specified by a set of parity equations, which signify that some subset of the symbols sum to zero, and can be represented by a bipartite graph with $n$ left vertices (symbols) and $n-k$ right nodes (parity equations); see Fig.~\ref{fig:ldpc}. The code has the following properties:
\begin{enumerate}
	\item Within the subclass of left- and right-regular LDPC codes, each parity equation is connected to a constant number $d_R$ of symbols, and each symbol is connected to a constant number $d_L$ of parity equations.
	\item The code can decoded via {\it peeling}, which has time complexity linear in the number of symbols:
	\begin{enumerate}
	    \item Starting with some set of known symbols, our goal is to recover the unknown symbols.
		\item Find a \textit{singleton}, or a parity equation containing only one unknown symbol. Compute the value of this symbol using the parity equation, and \textit{peel} this symbol by subtracting and removing it from each parity equation it participates in.
		\item Repeat until all the unknown symbols are  decoded (success), or there are no more singletons (failure).
	\end{enumerate}
	\item A \textit{stopping set} denotes a set of symbols that, if removed, prevent peeling decoding from succeeding. It is possible to construct an LDPC code such that the size of the smallest stopping set is a constant fraction $f$ of the total number of symbols.
\end{enumerate}

Next, we review the LDPC-coded Merkle tree~\cite{codedmerkle} (see Fig.~\ref{fig:tree}). The lowest level contains the $L$ transactions, which are coded into $2L$ symbols via a rate-$1/2$ LDPC code. Then, we take the hash of each symbol and group the hashes into groups of $4$ to produce the $L/2$ data symbols in the next layer. We then code these $L/2$ data symbols into $L$ symbols using an LDPC code, and so on. As in a typical Merkle tree, the Merkle proof for a symbol consists of its path to the root. A key property of this structure is that after decoding the $\ell$-th layer, the decoder has the hashes for the $(\ell+1)$-th layer, which will be exploited in the protocols below.

\begin{figure}[t]
	\centerline{\includegraphics[width=\linewidth]{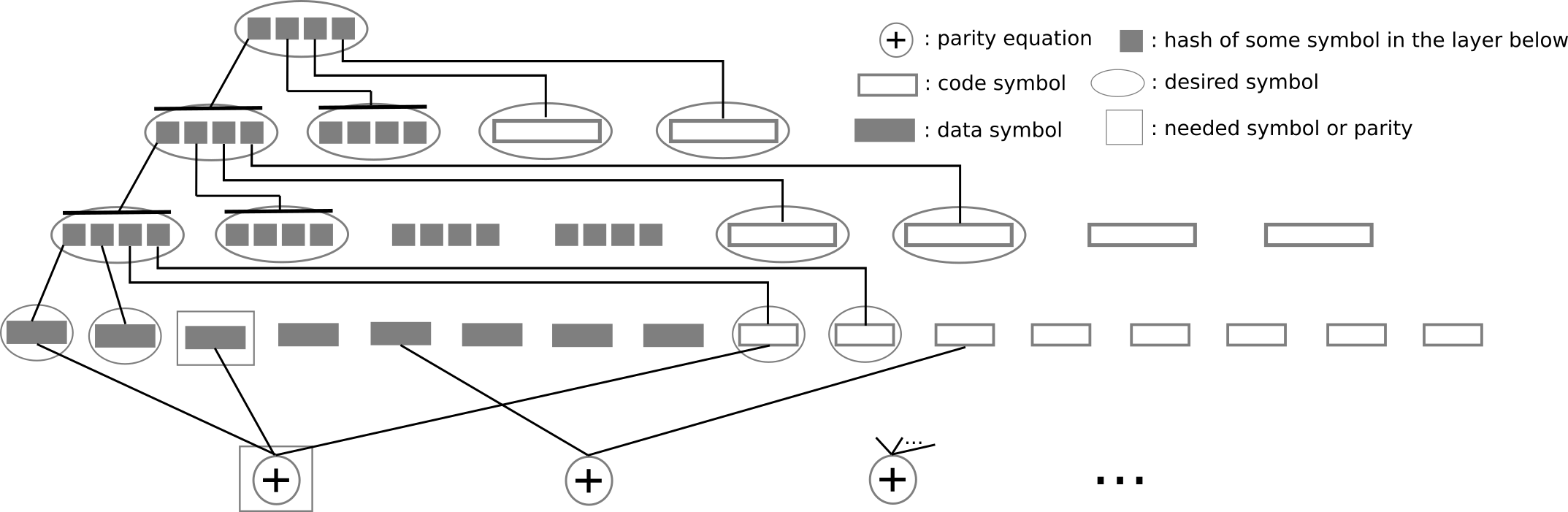}}
	\caption{The structure of a coded Merkle tree~\cite{codedmerkle}. Each row is encoded using a rate-$1/2$ code. Then, symbols are divided into groups of $4$, and their hashes are taken to produce the next level. The arrows denote a sub-tree, starting from the $4$ symbols sampled on the bottom layer. One possible sampled subtree is denoted by the circled symbols, and some example parities and symbols needed to decode this sampled subtree are boxed.}
	\label{fig:tree}
	\vspace{-5pt}
\end{figure}

\subsection{Checking Availability}\label{section:avail}
To make any symbol in the layer unavailable, the miner must prevent peeling decoding from succeeding by hiding a number of symbols that is at least the size of a stopping set, which by the LDPC code is at least a constant fraction $f$ of the data. Leveraging this idea, a light node can check for availability by randomly sampling $c$ symbols from each layer. If the miner has made a layer unavailable, then the node will sample an unavailable symbol with probability at least $1 - (1-f)^c$, which decays quickly in $c$. Therefore, each node samples $c$ symbols for each of the $\log L$ layers, along with their proofs to the root, and it accepts the data as available if and only if it receives all of its sampled data, resulting in $O(c (\log L)^2)$ total data downloaded. 

Using an interleaving approach in the tree structure, the amount of downloaded data can be reduced to $O(c \log L)$~\cite{codedmerkle}. Specifically, a node first chooses $c$ symbols on the bottom layer. Then, it also downloads their paths to the root, as well as any siblings of symbols on these paths. The arrows in Figure~\ref{fig:tree} denote one such sampled subtree. This subtree contains $c$ randomly chosen symbols from each layer, achieving the same sampling guarantees as above. Furthermore, because the sampled symbols form a subtree, each symbol's path to the root and therefore its Merkle proof is contained in the sampled data.

\begin{figure}[t]
	\centerline{\includegraphics[width=\linewidth]{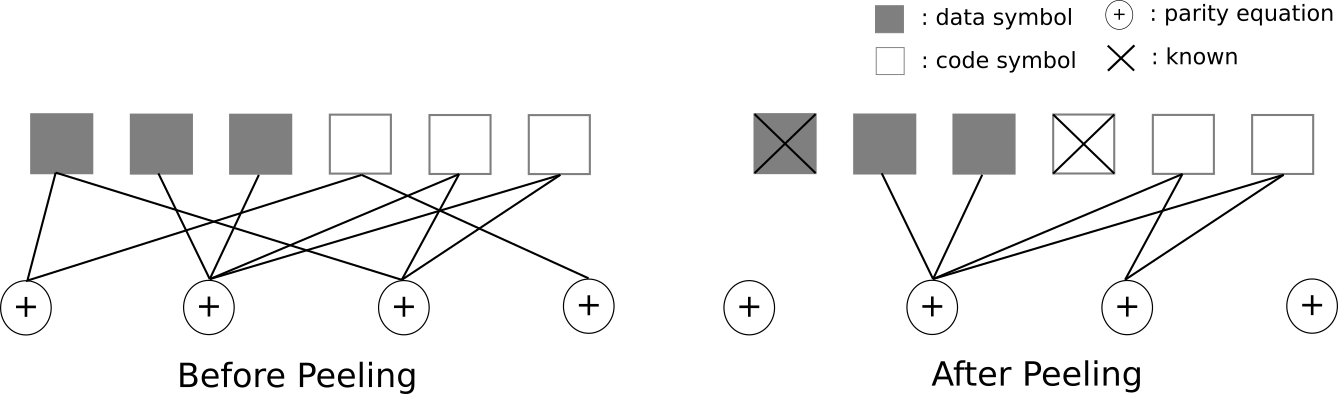}}
	\caption{A graph depiction of an LDPC code with $3$ data symbols and $3$ coded symbols, represented by gray and white boxes. Each $+$ circle denotes a parity equation, which signifies that the attached symbols sum to some known value, initially zero before peeling. After recursively peeling singletons (right), the first and fourth symbols are known.}
	\label{fig:ldpc}
	\vspace{-10pt}
\end{figure}

\subsection{Classical Decoding}
For the protocol above to function, however, the validators must have the capability to decode the block, otherwise the miner can hide small parts of the block without the fear that it will be decoded. We first review non-collaborative decoding, as described in \cite{codedmerkle}. First, the miner revealing some subset of the symbols along with Merkle proofs. Using these revealed symbols, the node attempts to decode the remaining symbols in the tree. We will decode layer by layer via peeling, starting from the top of the tree. For each layer, the decoding protocol should have the following properties:
\begin{enumerate}
    \item If the set of hidden symbols contains a stopping set, then peeling decoding will fail and the node will reject the block as unavailable.
    \item Otherwise, if there is coding fraud (at least one of the parity equations does not sum to zero), then the decoding procedure will produce a coding fraud proof.
    \item Otherwise, decoding is successful.
\end{enumerate}

The layered structure of the coded Merkle tree allows one to catch coding fraud. Suppose that the codes for layers $1,...,\ell-1$ are valid, but there is coding fraud in layer $\ell$. Because the node already decoded layer $\ell - 1$, it knows the hash of each unknown symbol in layer $\ell$. Then, if the node uses an invalid parity equation to decode a new symbol, the hash of this new decoded symbol will not match the known hash in layer $\ell -1$, allowing it to catch coding fraud. Using this idea, a node can decode a layer as follows \cite{codedmerkle}:
\begin{equation*}
    \texttt{decode\_layer}(\texttt{knownSymbols}, \texttt{hashes})
\end{equation*}
\begin{enumerate}[label=\texttt{\arabic*:}]
    \item While there are still unknown symbols,
    \begin{itemize}
        \item Look for a degree-one parity equation, i.e.\ an equation where all but one symbol are known.
        \item The unknown symbol is now known. Check that its hash matches the known hash in \texttt{hashes}. If the hash matches, add the symbol to \texttt{knownSymbols}.
        \item Otherwise, stop and produce a coding fraud proof.
        \item If there are no degree-one parity equations, stop and reject the block as unavailable.
    \end{itemize}
    \item Once everything is decoded, check the parity equations and make sure they sum to zero. If any equation does not, stop and produce a coding fraud proof.
    \item Otherwise, this layer is valid and has been decoded.
\end{enumerate}
Then, a node can decode the entire tree as follows:
\begin{equation*}
    \texttt{decode\_tree}()
\end{equation*}
\begin{enumerate}[label=\texttt{\arabic*:}]
    \item For each layer $\ell$, keep track $\texttt{knownSymbols}(\ell)$ and their hashes $\texttt{hashes}(\ell)$. Download as many symbols as possible, and add each symbol with a valid Merkle proof to $\texttt{knownSymbols}(\ell$).
    \item Initialize $\texttt{hashes}(1)$ to the \texttt{root} of the Merkle tree.
    \item For layer $\ell = 1,2,...,\lceil \log L \rceil$:
    \begin{itemize}
        \item Call $\texttt{decode\_layer}(\texttt{knownSymbols}(\ell),$
        $\texttt{hashes}({\ell}))$, which produces $\texttt{hashes}({\ell}+1)$.
    \end{itemize}
\end{enumerate}

\subsection{Coding Fraud Proofs}\label{section:codingfraudproof}
Suppose that after peeling some degree-one parity equation using symbols $s_1,...s_{d-1}$, the decoded symbol $\hat{s}_d$ has an incorrect hash. Because this parity equation was the first to fail, the $d-1$ previously known symbols and the $d$th hash are all valid, known, and have Merkle proofs to the root. Then, a \texttt{codingFraudProof} contains the following:
\begin{itemize}
    \item $\hat{s}_d$: the newly decoded incorrect symbol.
	\item $h_d$: the committed hash of the $d$-th symbol.
	\item \texttt{hashProof}: the Merkle proof of the hash $h_d$.
	\item $s_1,s_2,...,s_{d-1}$: the $d-1$ known symbols corresponding to the parity-check equation for $\hat{s}_d$.
	\item \texttt{symbolProofs}: Merkle proofs of $s_1,s_2,...,s_{d-1}$.
\end{itemize}
\begin{align*}
    \texttt{is\_valid\_codingFraudProof}( \qquad \qquad \qquad \qquad \\
    \qquad \qquad\texttt{codingFraudProof}) \in \{\texttt{True}, \texttt{False}\}
\end{align*}
\begin{enumerate}[label=\texttt{\arabic*:}]
	\item The \texttt{hashProof} and \texttt{symbolProofs} are correct.
	\item There is a parity equation involving these $d$ symbols.
	\item Decoding this parity equation results in the $\hat{s}_d$ such that $0 = \hat{s}_d + \sum_{i=1}^{d-1} s_i$.
	\item The commitment is invalid, or $\texttt{hash}\left(\hat{s}_d\right) \neq h_d$.
	\item If any of the above checks are false, return \texttt{False}. Otherwise, return \texttt{True}.
\end{enumerate}
The coding fraud proof is size $O(d_R \log L)$ from the $d_R$ Merkle proofs, where $d_R$ (LDPC right degree) is constant.

\subsection{Collaborative Decoding}\label{section:collabdecode}\label{secion:correctnessproof}
To enable this protocol for light nodes, we propose a collaborative decoding scheme. At a high level, decoding follows the steps described above, but each node decodes only a subtree of the entire tree. Every time a node decodes a new symbol, it broadcasts it, allowing other nodes to use it. Each broadcasted symbol must be accompanied by a Merkle proof so that malicious nodes cannot broadcast fake symbols. Decoding succeeds as long as together, the honest nodes cover the entire tree. 

First, we describe the subtree version of $\texttt{decode\_layer}$, assuming the previous layer in the subtree has already been decoded. Decoding layer $\ell$ involves the following sets:
\begin{itemize}
    \item $\texttt{desiredSymbols}(\ell)$: the layer $\ell$ symbols in the subtree. Note that their Merkle proofs are contained in previous layers of the subtree.
    \item $\texttt{knownDesiredSymbols}(\ell)$: the subset of the desired symbols that have been successfully decoded. 
    \item $\texttt{hashes}(\ell)$: the hashes of the desired symbols, which are in $\texttt{knownDesiredSymbols}({\ell-1})$, or the previous layer of the subtree.
    \item $\texttt{neededParities}(\ell)$: any parity equations attached to at least one unknown desired symbol.
    \item $\texttt{neededSymbols}(\ell)$: any symbol attached to a needed parity equation.\footnote{Note that because the left and right degree of the LDPC code are at most $d_L$ and $d_R$, we have $|\texttt{neededSymbols}| \leq d_Rd_L\ |\texttt{desiredSymbols}|$, so the size of this set is not too big.}
    \item $\texttt{knownNeededSymbols}_\ell$: the subset of needed symbols that are known, along with their Merkle proofs.
\end{itemize}
At a high level, each node's job is to securely decode the symbols in $\texttt{desiredSymbols}(\ell)$, and it has stored their hashes after decoding the previous layer. To accomplish this job, the node needs to decode parities, so it downloads the symbols in $\texttt{neededSymbols}(\ell)$, depending on other nodes to decode them. An example is shown in Figure~\ref{fig:tree}. To make the protocol secure, any downloaded symbols must be accompanied by a Merkle proof, and the node also broadcasts any symbol it decodes with an Merkle proof. Specifically, decoding a layer proceeds as follows:
\begin{align*}
    \texttt{collab\_decode\_layer}
    (&\texttt{hashes},\\
    &\texttt{knownDesiredSymbols}, \\
    &\texttt{knownNeededSymbols})
\end{align*}
\begin{enumerate}[label=\texttt{\arabic*:}]
    \item While there are still unknown desired symbols,
    \begin{itemize}
        \item Look for a degree-one needed parity equation (all but one symbol are in $\texttt{knownDesiredSymbols}$ or $\texttt{knownNeededSymbols}$).
        \item The newly decoded symbol is a desired symbol. Check that its hash matches the known hash in \texttt{hashes}. If the hash matches, add the symbol to \texttt{knownDesiredSymbols} and broadcast it with a Merkle proof. Otherwise, stop and produce a \texttt{codingFraudProof} (see Section~\ref{section:codingfraudproof}).
        \item Listen for any broadcasts of \texttt{neededSymbols}. Check their Merkle proofs. If the Merkle proof is valid, add the symbol and its proof to \texttt{knownNeededSymbols}. Store the Merkle proof in case it is needed for a \texttt{codingFraudProof}.
        \item If there are no degree-one parity equations, wait some fixed delay. If there were no additional needed symbols broadcasted in this delay, stop and reject the block as unavailable.
    \end{itemize}
    \item Once everything is decoded, check degree-zero needed parity equations and make sure they sum to zero. If any equation does not, produce a coding fraud proof.
    \item Otherwise, this layer is valid and has been decoded.
\end{enumerate}
We can decode an entire subtree by decoding each layer:
\begin{equation*}
    \texttt{decode\_subtree}()
\end{equation*}
\begin{enumerate}[label=\texttt{\arabic*:}]
    \item Choose a subtree of $O(c\log L)$ symbols, as described in Section~\ref{section:avail} and Figure~\ref{fig:tree}. This subtree defines the $\texttt{desiredSymbols}(\ell)$, $\texttt{neededParities}(\ell)$, and $\texttt{neededSymbols}(\ell)$ for each layer $\ell$.
    \item Attempt to download any desired or needed symbols from the network, along with their Merkle proofs. Add each symbol with a valid proof to $\texttt{knownDesiredSymbols}(\ell)$ or $\texttt{knownNeededSymbols}(\ell)$, where $\ell$ is the layer that the symbol belongs to.
    \item For layer $\ell = 1,2,...,\lceil \log L \rceil$,
    \begin{itemize}
        \item Decode using
        \begin{align*}
            &\texttt{collab\_decode\_layer}(\\
            &\qquad\texttt{knownDesiredSymbols}_\ell, \\
            &\qquad\texttt{knownDesiredSymbols}_{\ell-1}, \\ 
            &\qquad\texttt{knownNeededSymbols}_\ell).
        \end{align*}
        \item If decoding was unsuccessful, stop.
        \item Also listen for any \texttt{codingFraudProof}. If any is valid, stop and reject the block.
    \end{itemize}
\end{enumerate}
Using similar calculations as Section~\ref{section:coverage}, if there are at least
\begin{equation*}
N_h > \frac{L}{c}(\log L + \lambda)
\end{equation*}
honest nodes, then the entire tree is covered with probability at least $(1-e^{-\lambda})$. Note that, as for validation, the fraction $\frac{c}{L}$ of work required scales down proportionally with $N_h$ up to a log factor. For simplicity, we can set $c = L/k$, such that $1/k$-th of the block is sampled for validation and another $1/k$-th is sampled for availability.

\section{Selective Broadcast}\label{section:broadcast}
Finally, because light nodes cannot download the entire block, we adapt the conventional gossip protocol to enable selective downloading. The conventional gossip protocol, which we use for block headers and any fraud proofs, proceeds as follows (see~\cite{Gossip:survey:07} for more details):
\begin{enumerate}
    \item The starting node sends the message to all neighbors.
    \item Each neighboring node does the following: if the node has not already seen the message, it first validates the message (e.g., checks if the fraud proof is valid). If the message is valid, the node forwards the message to all of its neighbors. If the message is invalid or has been seen, the node ignores it.
\end{enumerate}
Note that fake communication from malicious nodes are mitigated by honest nodes only gossiping valid messages.

Other broadcasts, like ones involving symbols in the tree, are only of interest to a subset of the nodes. To reduce communication, we propose and analyze a simple modification of the gossip protocol, called {\it selective broadcast}:
\begin{enumerate}
	\item Each node chooses the symbols it is interested in and informs its neighbors of these symbols.
	\item A node only gossips each received symbol to neighbors that are also interested in that symbol.
\end{enumerate}
Then, the communication complexity per node is simply the total size of the data it wishes to download, plus the initial step of informing neighbors of desired symbols, which has complexity linear in the number of neighbors. The protocol is successful, meaning that every node receives the data it wishes to download, as long as for each symbol, the subgraph of nodes interested in that symbol is connected.

This property always holds if the network graph is fully connected. While nodes would need many neighbors, the only communication that scales with number of neighbors is the initial information step. Because these messages are relatively small, this communication is small compared to the work needed to download, decode, and validate block data, especially if the number of participants is small. Therefore, it is reasonable with a small group of light nodes.

If the number of participants is large, then we can reduce the connectivity. Under the simplifying assumption that the graph is Erd\H{o}s-R{\'e}nyi \cite{graph}, we analyze the required connectivity in Section~\ref{section:connectivity} and find that each node should have at least $O\left( k \log N_h \right)$
neighbors, where $1/k$ is the fraction of each block that each node downloads. In the case where there are malicious nodes, these nodes can choose to ignore all messages, increasing the amount of connectivity needed. In this case, each honest node will need $O\left( \frac{1}{1-\alpha} k \log N_h \right)$ neighbors, where $\alpha$ is the fraction of malicious nodes. This term, which represents the amount of communication during the information step, is linear in $k$, which is exactly the factor of work reduced. Therefore, choosing a savings factor of $k = O(\sqrt{L})$ results in $O(\sqrt{L} (\log L + \log N_h))$ computation, bandwidth, and storage.

\section{Analysis of CoVer}\label{section:connectivity}\label{section:correctnessproof}\label{section:collabdecodeanalysis}\label{section:coverage}
In this section, we formally analyze CoVer; please see the appendix for full proofs.
Let each of the $N_h$ honest nodes choose a section out of $k$ at random to validate, and a random $1/k$th to sample for availability. Let $\lambda$ denote the security parameter such that the protocol succeeds with probability $\geq 1 - e^{-O(\lambda)}$. First, we calculate the number of nodes needed to cover every section with probability $\geq 1-e^{-\lambda}$ and find that we need $N_h \geq k (\log k + \lambda)$. Next, given that the block is unavailable, all honest nodes sample at least one unavailable symbol with probability $\geq 1 - N_h (1-f)^{L/k}$. Finally, we calculate connectivity required for selective broadcast and find that if each node has $O(k \lambda \log N_h)$ neighbors, the success probability is at least $\geq 1-  4L \left(\frac{N_h}{8k}\right)^{1-\lambda} - 4L \exp \left( - \frac{N_h}{8(4k-1)} \right)$. Putting these results together via the union bound, we analyze correctness and show that with high probability, all light nodes accept the block if and only if all of its data is available for download, and every included transaction is valid.

\begin{thm}
    CoVer satisfies the following:
    \begin{enumerate}
        \item If the block is valid and available, then all honest participants will accept the block with probability $\geq 1 - e^{-\lambda} -  4L \left(\frac{N_h}{8k}\right)^{1-\lambda} - 4L \exp \left( - \frac{N_h}{8(4k-1)} \right)$.
        \item If the block is unavailable, then all honest participants will reject  with probability $\geq 1 - N_h (1-f)^{L/k}$.
        \item If the block is invalid but available, then all honest participants will reject the block with probability $\geq 1 - e^{-\lambda} -  4L \left(\frac{N_h}{8k}\right)^{1-\lambda} - 4L \exp \left( - \frac{N_h}{8(4k-1)} \right)$.
    \end{enumerate}
\end{thm}

\section{Related Work}
\label{section:related-work}
To the best of our knowledge, CoVer is the first light-node-only protocol for block validation. First, we briefly survey works related to light node protocols. Light nodes, or simplified payment verification (SPV) clients, were proposed by Nakamoto in the original Bitcoin paper~\cite{bitcoin}. These nodes only store block headers and perform no validation, relying on full nodes for transaction membership proofs. As a result, they are vulnerable to several privacy and security attacks~\cite{Karame:16}, and several works are aimed at improving their privacy and security; see, e.g.,~\cite{Bloom:SPV:12,Client-Bloom:SPV:17,Bite:SPV:19}. Rather than improving security, another line of work instead reduces the light node computational burden even further by enabling them to download only a sublinear (in the length of the blockchain) number of block headers~\cite{Kiayias:16:liteSPV,Kiayias:20:liteSPV,Flyclient:20:liteSPV}.

Several works focus on reducing the costs associated with running a full node, especially storage costs~\cite{Raman:storage:17,Perard:storage:18,storage} and communication cost during bootstrap~\cite{vault}. For sharded blockchains, Polyshard~\cite{Polyshard:20} proposes a protocol for verification functions that can be represented as polynomials.

Our work directly builds on top of recent work on fraud proofs and data availability. In \cite{fraud}, Al-Bassam et al.\ propose the ideas of fraud proofs and data availability, allowing light nodes to receive compact proofs of block invalidity from full nodes. While Al-Bassam et al.\ use 2D Reed Solomon codes for data availability, Yu et al.~\cite{codedmerkle} propose a new coding scheme called the coded Merkle tree to further reduce the coding fraud proof sizes and decoding complexity. While these works require light nodes to rely on full nodes, our main contribution is to remove this requirement through light-node-only protocols.

\section{Conclusion}
\label{section:conclusion}
We propose CoVer, a decentralized and permissionless protocol that enables light nodes to collaboratively validate blocks even under a dishonest majority, achieving the same level of security as full nodes with a fraction of the work. CoVer allows nodes of all sizes to contribute to network security, bringing us closer to fully scalable blockchains and enabling a new world of possibilities.

\bibliographystyle{IEEEtran}
\bibliography{Bib_CoVer}

\begin{thebibliography}{10}
\providecommand{\url}[1]{#1}
\csname url@samestyle\endcsname
\providecommand{\newblock}{\relax}
\providecommand{\bibinfo}[2]{#2}
\providecommand{\BIBentrySTDinterwordspacing}{\spaceskip=0pt\relax}
\providecommand{\BIBentryALTinterwordstretchfactor}{4}
\providecommand{\BIBentryALTinterwordspacing}{\spaceskip=\fontdimen2\font plus
\BIBentryALTinterwordstretchfactor\fontdimen3\font minus
  \fontdimen4\font\relax}
\providecommand{\BIBforeignlanguage}[2]{{%
\expandafter\ifx\csname l@#1\endcsname\relax
\typeout{** WARNING: IEEEtran.bst: No hyphenation pattern has been}%
\typeout{** loaded for the language `#1'. Using the pattern for}%
\typeout{** the default language instead.}%
\else
\language=\csname l@#1\endcsname
\fi
#2}}
\providecommand{\BIBdecl}{\relax}
\BIBdecl

\bibitem{bitcoin}
\BIBentryALTinterwordspacing
S.~Nakamoto, ``Bitcoin: A peer-to-peer electronic cash system,'' 2009.
  [Online]. Available: \url{http://www.bitcoin.org/bitcoin.pdf}
\BIBentrySTDinterwordspacing

\bibitem{Karame:16}
G.~Karame and E.~Audroulaki, \emph{Bitcoin and Blockchain Security}.\hskip 1em
  plus 0.5em minus 0.4em\relax Norwood, MA, USA: Artech House, Inc., 2016.

\bibitem{fraud}
\BIBentryALTinterwordspacing
M.~Al{-}Bassam, A.~Sonnino, and V.~Buterin, ``Fraud proofs: Maximising light
  client security and scaling blockchains with dishonest majorities,''
  \emph{CoRR}, vol. abs/1809.09044, 2018. [Online]. Available:
  \url{http://arxiv.org/abs/1809.09044}
\BIBentrySTDinterwordspacing

\bibitem{codedmerkle}
\BIBentryALTinterwordspacing
M.~Yu, S.~Sahraei, S.~Li, S.~Avestimehr, S.~Kannan, and P.~Viswanath, ``Coded
  merkle tree: Solving data availability attacks in blockchains,'' in
  \emph{Financial Cryptography and Data Security (FC)}, 2020. [Online].
  Available: \url{https://eprint.iacr.org/2019/1139.pdf}
\BIBentrySTDinterwordspacing

\bibitem{graph}
P.~Erd\H{o}s and A.~R{\'e}nyi, ``On the evolution of random graphs,''
  \emph{Publication of the Mathematical Institute of the Hungarian Academy of
  Sciences}, 1960.

\bibitem{vault}
D.~Leung, A.~Suhl, Y.~Gilad, and N.~Zeldovich, ``Vault: Fast bootstrapping for
  the algorand cryptocurrency,'' in \emph{26th Annual Network and Distributed
  System Security Symposium, {NDSS}}, 2019.

\bibitem{ldpc}
R.~{Gallager}, ``Low-density parity-check codes,'' \emph{IRE Transactions on
  Information Theory}, vol.~8, no.~1, pp. 21--28, 1962.

\bibitem{Richardson:08}
T.~Richardson and R.~Urbanke, \emph{Modern Coding Theory}.\hskip 1em plus 0.5em
  minus 0.4em\relax New York, NY, USA: Cambridge University Press, 2008.

\bibitem{Gossip:survey:07}
K.~Birman, ``The promise, and limitations, of gossip protocols,'' \emph{SIGOPS
  Oper. Syst. Rev.}, p. 8–13, Oct. 2007.

\bibitem{Bloom:SPV:12}
M.~Hearn and M.~Corallo, ``Connection bloom filtering,'' Bitcoin Improvement
  Proposal 37, 2012, \url{https://github.com/bitcoin/bips/blob/master/ bip-
  0037.mediawiki}.

\bibitem{Client-Bloom:SPV:17}
O.~Osuntokun, A.~Akselrod, and J.~Posen, ``Client side bloom filtering,''
  Bitcoin Improvement Proposal 157, 2017,
  \url{https://github.com/bitcoin/bips/blob/ master/bip-0157.mediawiki}.

\bibitem{Bite:SPV:19}
S.~Matetic, K.~W{\"u}st, M.~Schneider, K.~Kostiainen, G.~Karame, and S.~Capkun,
  ``{BITE}: Bitcoin lightweight client privacy using trusted execution,'' in
  \emph{28th {USENIX} Security Symposium ({USENIX} Security 19)}, Aug. 2019.

\bibitem{Kiayias:16:liteSPV}
A.~Kiayias, N.~Lamprou, and A.-P. Stouka, ``Proofs of proofs of work with
  sublinear complexity,'' in \emph{Financial Cryptography Workshops}, 2016.

\bibitem{Kiayias:20:liteSPV}
\BIBentryALTinterwordspacing
A.~Kiayias, A.~Miller, and D.~Zindros, ``Non-interactive proofs of
  proof-of-work,'' in \emph{Financial Cryptography and Data Security (FC)},
  2020. [Online]. Available: \url{https://eprint.iacr.org/2017/963.pdf}
\BIBentrySTDinterwordspacing

\bibitem{Flyclient:20:liteSPV}
B.~Bu\"{u}\"{n}z, L.~Kiffer, L.~Luu, and M.~Zamani, ``{FlyClient}: Super-light
  clients for cryptocurrencies,'' in \emph{2020 IEEE Symposium on Security and
  Privacy (SP)}, 2020.

\bibitem{Raman:storage:17}
\BIBentryALTinterwordspacing
R.~K. Raman and L.~R. Varshney, ``Dynamic distributed storage for scaling
  blockchains,'' \emph{CoRR}, vol. abs/1711.07617, 2017. [Online]. Available:
  \url{http://arxiv.org/abs/1711.07617}
\BIBentrySTDinterwordspacing

\bibitem{Perard:storage:18}
\BIBentryALTinterwordspacing
D.~Perard, J.~Lacan, Y.~Bachy, and J.~Detchart, ``Erasure code-based low
  storage blockchain node,'' \emph{CoRR}, vol. abs/1805.00860, 2018. [Online].
  Available: \url{http://arxiv.org/abs/1805.00860}
\BIBentrySTDinterwordspacing

\bibitem{storage}
\BIBentryALTinterwordspacing
S.~Kadhe, J.~Chung, and K.~Ramchandran, ``{SeF}: A secure fountain architecture
  for slashing storage costs in blockchains,'' in \emph{Scaling Bitcoin}, Sep.
  2019. [Online]. Available: \url{https://arxiv.org/pdf/1906.12140}
\BIBentrySTDinterwordspacing

\bibitem{Polyshard:20}
S.~{Li}, M.~{Yu}, C.~{Yang}, A.~S. {Avestimehr}, S.~{Kannan}, and
  P.~{Viswanath}, ``Polyshard: Coded sharding achieves linearly scaling
  efficiency and security simultaneously,'' \emph{IEEE Transactions on
  Information Forensics and Security}, vol.~16, pp. 249--261, 2020.

\end{thebibliography}

\section{Appendix}

\subsection{LDPC Code Example} \label{section:ldpc-example}

Consider a toy example for an LDPC code represented using a graph shown in Fig.~\ref{fig:ldpc-toy}. The left nodes (called variable nodes) correspond to codeword symbols, and the right nodes (called check nodes) correspond to parity-check equations. Suppose $X_0, X_2$, and $X_4$ are erased. Let us see how they can be recovered using the peeling decoder. It is an iterative decoder. In each iteration, the decoder first finds a check node that has only one neighbor erased. Then, it recovers the value of the erased neighbor. For instance, in the first iteration, the decoder finds that the first check node $X_0+X_2+X_4$ has only $X_4$ erased. It computes $X_4 = -(X_0+X_2)$. We illustrate the process in Fig.~\ref{fig:ldpc-peeling}.

\begin{figure}[!t]
    \centering
    \includegraphics[width=0.5\linewidth]{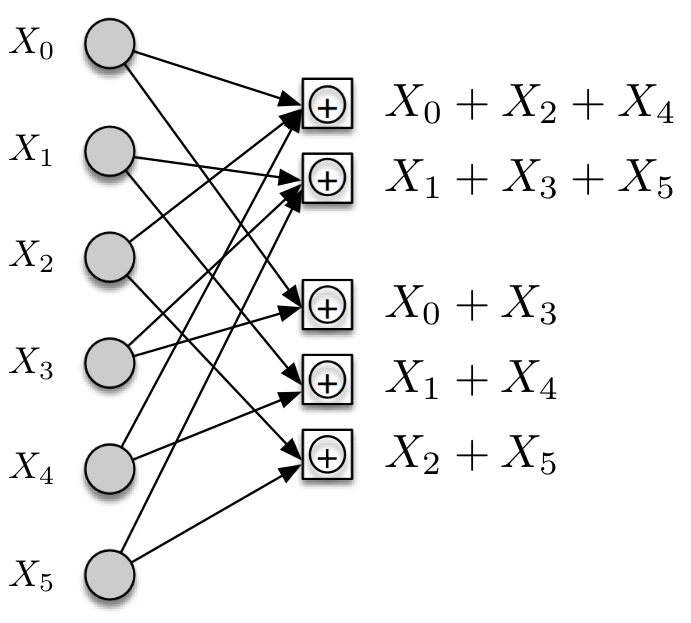}
    \caption{A toy example of a graph associated with an LDPC code.}
    \label{fig:ldpc-toy}
\end{figure}

\begin{figure}[!t]
    \centering
    \includegraphics[width=\linewidth]{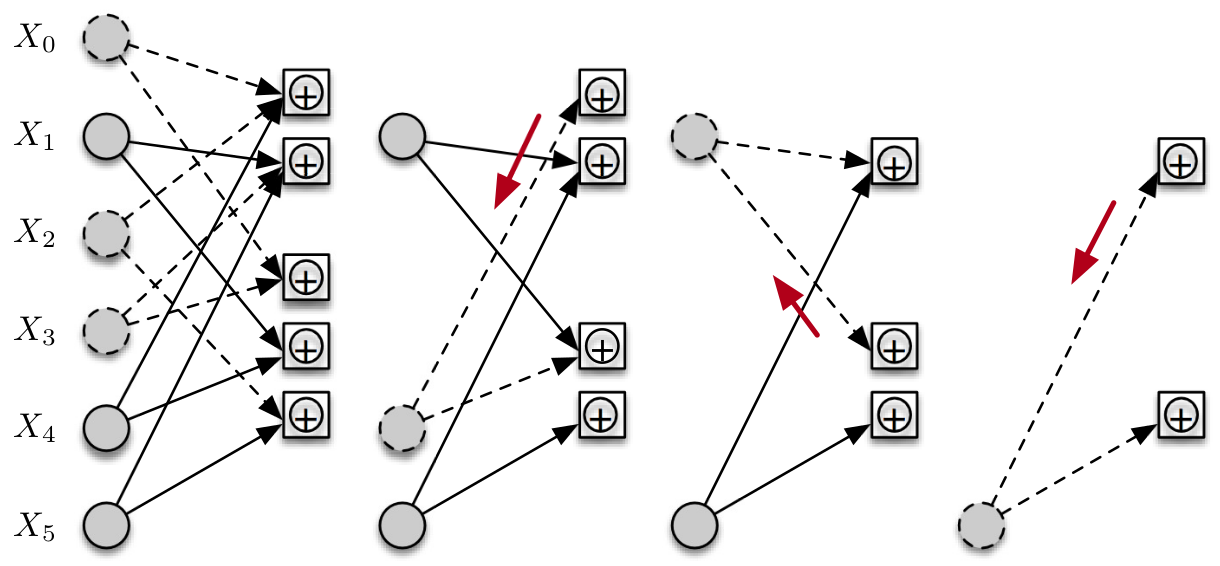}
    \caption{A toy example of the peeling decoder.}
    \label{fig:ldpc-peeling}
\end{figure}

\subsection{Collaborative Coding Fraud Proof Example}\label{section:codingfraudexmp}

Next, we provide an example of a coding fraud proof produced during collaborative decoding. Decoding proceeds layer by layer, starting from the top. Suppose we have the following tree, where a light node is trying to decode the third layer (layers below not shown):
\\\\
\centerline{\includegraphics[width=\linewidth]{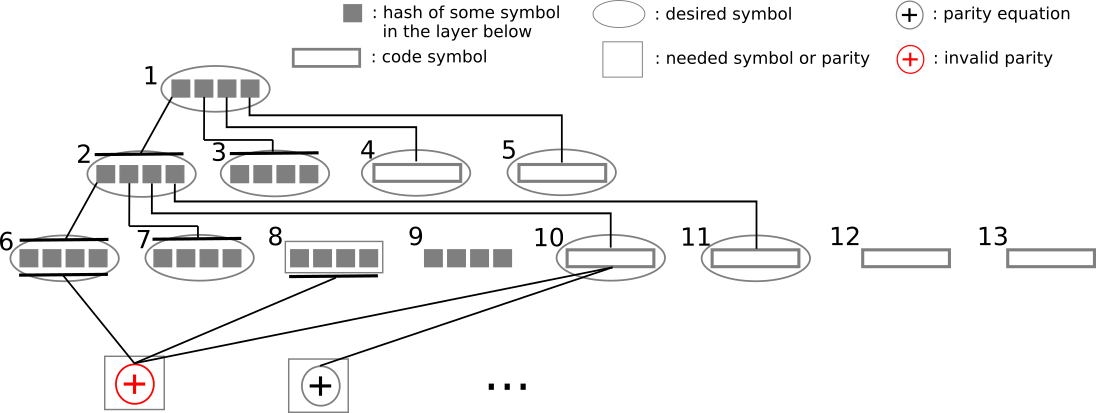}}
\\\\

Currently, the node is attempting to decode using the parities (two of which are shown), but the one involving symbols 6, 8, and 10 is invalid because the miner constructed the code incorrectly. In other words, the parity does not sum to zero, but the node is unaware of this error. 

Note that because we are on layer 3, the node has stored symbols 1, 2, 3 (the coded symbols can be discarded after the layer is decoded). If this layer is decoded successfully, then it will add symbols 6 and 7 to its storage. Decoding will proceed as follows:
\begin{enumerate}
    \item First, the node will peel symbol 10, check that its hash matches the third hash in symbol 2, and then broadcast it with a Merkle proof (which consists of symbols 10, 2, and 1). 
    \item Next, it will wait for symbols to be broadcasted with a Merkle proof. Suppose symbol 8 is broadcasted and the Merkle proof is valid.
    \item Then, the node will attempt to use the red parity to decode symbol 6. It will find that the resulting hash does not match the first hash in symbol 2, so the red parity is invalid.
    \item The node will produce and broadcast a coding fraud proof containing the following:
    \begin{itemize}
        \item $\hat{s}_6$: the newly decoded symbol 6
        \item $h_6$: the hash of symbol 6, which is the first quarter of symbol 2.
        \item $\texttt{hashProof}$: the proof of $h_6$, consisting of symbols $2$ and $1$.
        \item $s_8, s_{10}$: the other symbols in the parity, or symbols 8 and 10.
        \item $\texttt{symbolProofs}$: the Merkle proofs, or symbols 8, 3, and 1 and symbols 10, 2, and 1.
    \end{itemize}
\end{enumerate}
Other nodes will check the fraud proof, and find that all of its Merkle proofs are correct and $s_8 + s_{10} + \hat{s}_6 = 0$, but $\texttt{hash}(\hat{s}_6) \neq h_6$, so they will reject the block.

\subsection{Transaction Fraud Proof Example}\label{section:valfraudexmp}
In this section, we provide a simple example of an invalid transaction and the resulting fraud proof. Suppose we have the following sequence of transactions:
\\\\
\begin{tabular}{l|l}
    \multicolumn{2}{l}{Transaction 5 in block 8} \\
    \midrule
    \texttt{txid}  & 8:5 \\
    \texttt{sender} & 10 \\
    \texttt{signature} & \texttt{sign}(10) \\
    \texttt{outputs} & 1.\ \$10 to account 8 \\
    & 2.\ \$8 to account 3\\
    \texttt{inputs} & \$18 from output 3 in \texttt{txn}$_{5:2}$ \\
    \texttt{inputProofs} & $\texttt{proof}(\texttt{txn}_{5:2} \to \texttt{header}(5).\texttt{root})$ \\
    \bottomrule
\end{tabular}
\\\\
$\;$
\\\\
\begin{tabular}{l|l}
    \multicolumn{2}{l}{Transaction 3 in block 9} \\
    \midrule
    \texttt{txid}  & 9:3 \\
    \texttt{sender} & 8 \\
    \texttt{signature} & \texttt{sign}(8) \\
    \texttt{outputs} & 1.\ \$5 to account 4 \\
    & 2.\ \$5 to account 8\\
    \texttt{inputs} & \$10 from output 1 in \texttt{txn}$_{8:5}$ \\
    \texttt{inputProofs} & $\texttt{proof}(\texttt{txn}_{8:5} \to \texttt{header}(8).\texttt{root})$ \\
    \bottomrule
\end{tabular}
\\\\
$\;$
\\\\
\begin{tabular}{l|l}
    \multicolumn{2}{l}{Transaction 2 in block 10} \\
    \midrule
    \texttt{txid}  & 10:2 \\
    \texttt{sender} & 8 \\
    \texttt{signature} & \texttt{sign}(8) \\
    \texttt{outputs} & 1.\ \$10 to account 3 \\
    \texttt{inputs} & \$10 from output 1 in \texttt{txn}$_{8:5}$ \\
    \texttt{inputProofs} & $\texttt{proof}(\texttt{txn}_{8:5} \to \texttt{header}(8).\texttt{root})$ \\
    \bottomrule
\end{tabular}
\\\\
$\;$
\\\\
In this case, block 10 in invalid because sender 8 is trying to double-spend output 1 from transaction 8:5. Then, when block 10 is broadcast, the light node validating account 8 will catch the double spending attempt because it will have transaction 9:3 in its stored history. Then, it will produce and broadcast the following fraud proof. Any recipient of this fraud proof will check that the two transactions' inputs collide and therefore reject block 10.
\\\\
\begin{tabular}{l}
     Fraud proof for transaction 2 in block 10 \\
     \midrule
     \resizebox{\linewidth}{!}{\begin{tabular}{l|l}
    \texttt{invalidTxn}  & \begin{tabular}{l|l}
    \texttt{txid}  & 10:2 \\
    \texttt{sender} & 8 \\
    \texttt{signature} & \texttt{sign}(8) \\
    \texttt{outputs} & 1.\ \$10 to account 3 \\
    \texttt{inputs} & \$10 from output 1 in \texttt{txn}$_{8:5}$ \\
    \texttt{inputProofs} & $\texttt{proof}(\texttt{txn}_{8:5} \to \texttt{header}(8).\texttt{root})$ \\
\end{tabular} \\
\midrule
    \texttt{invalidTxnProof} & $\texttt{proof}(\texttt{txn}_{10:2} \to \texttt{header}(10).\texttt{root})$ \\
\midrule
    \texttt{pastTxn} & \begin{tabular}{l|l}
    \texttt{txid}  & 9:3 \\
    \texttt{sender} & 8 \\
    \texttt{signature} & \texttt{sign}(8) \\
    \texttt{outputs} & 1.\ \$5 to account 4 \\
    & 2.\ \$5 to account 8\\
    \texttt{inputs} & \$10 from output 1 in \texttt{txn}$_{8:5}$ \\
    \texttt{inputProofs} & $\texttt{proof}(\texttt{txn}_{8:5} \to \texttt{header}(8).\texttt{root})$ \\
\end{tabular} \\
\midrule
    \texttt{pastTxnProof} & $\texttt{proof}(\texttt{txn}_{8:5} \to \texttt{header}(8).\texttt{root})$   \\
    \bottomrule
\end{tabular}}
\end{tabular}

\subsection{Collaborative Validation Coverage Analysis}
Suppose that a block is divided into $k$ sections, with each honest node validating a single section. We wish to analyze the probability that together, the honest nodes cover the block.
\begin{thm}\label{thm:coverage}
	Suppose there are $k$ sections and each of the $N_h$ participants chooses a section at random. Then, if $N_h \geq k (\log k + \lambda)$, all of the sections are covered with probability at least $1-e^{-\lambda}$.
\end{thm}
\begin{proof}
	By the union bound,
	\begin{align*}
	&\Pr(\text{all sections covered}) \\
	&\geq 1- k \left( 1 - \frac{1}{k} \right)^{N_h} \\
	&\geq 1 - ke^{-\frac{1}{k}N_h}.
	\end{align*}
	Choosing $N_h \geq k (\log k + \lambda)$ results in the bound $1 - e^{-\lambda}$.
\end{proof}

\subsection{Selective Broadcast Connectivity Analysis}
We would like to analyze the probability that the selective broadcast procedure succeeds, meaning that every node receives the data it wishes to download. We will represent the connectivity between nodes as a graph.

Note that if a block has $L$ transactions, then the coded Merkle tree contains $2L + L + L/2 + ... + 1 \approx 4L$ symbols. Let $r = \frac{(L/k)\log L}{4L}$ denote the fraction of total symbols that each node downloads, and let $M$ denote the total number of symbols. 

Then, the procedure succeeds if for each of the $M$ symbols, the subgraph containing nodes interested in that symbol is connected. The following theorem analyzes the connectivity required for selective broadcast to succeed under the simplifying assumption that the connectivity graph is an Erd\H{o}s-R{\'e}nyi random graph \cite{graph}.
\begin{thm}\label{thm:connectivity}
	Suppose we have a graph $G(N_h, p)$ of nodes, where each node is connected independently with probability $p$. Also, suppose each node is colored with $rM$ of $M$ colors, chosen uniformly at random. Then, for $p \geq \frac{2 \lambda \log (rN_h/2)}{rN_h}$, all of the color subgraphs are connected with probability at least $1 - M (rN_h/2)^{1-\lambda} - M\exp \left( - \frac{r N_h}{8(1-r)} \right)$.
\end{thm}
\begin{proof}
	Note that each subgraph is also Erd\H{o}s-R{\'e}nyi with connection probability $p$. Then, the proof will proceed as follows: first, we will produce a high-probability lower bound on the number of nodes per subgraph. Letting $n$ denote this lower bound, we will then set $p \geq \frac{\lambda \log n}{n}$, which by standard $G(n,p)$ arguments means that $\Pr(\text{subgraph disconnected}) \leq n^{1-\lambda}$. Finally, we will use the union bound over subgraphs.
	
	Note that the number of nodes per subgraph is $\text{Binom}(N_h, r)$, but the subgraphs are not independent of each other. Then, to produce the lower bound on nodes per subgraph, we will use the union bound combined with the a sub-Gaussian bound for the Binomial:
	\begin{align*}
	&\Pr(\text{for all subgraphs, \# nodes} \geq rN_h/2) \\ 
	&\geq 1 - M \Pr( \text{Binom}(N_h, r) \leq rN_h/2) \\
	&\geq 1 - M \exp \left( - \frac{r^2 N_h^2/4}{2 N_h r (1-r)} \right) \\
	&= 1 - M \exp \left( - \frac{r N_h}{8(1-r)} \right)
	\end{align*}
	Then, setting 
	\begin{equation*}
	p \geq \frac{2 \lambda \log (rN_h/2)}{rN_h},
	\end{equation*}
	the probability that all subgraphs are connected conditioned on the high-probability lower bound is
	\begin{align*}
	&\Pr(\text{all subgraphs connected} \mid \text{\# nodes} \geq rN_h/2 ) \\
	&\geq 1 - M \Pr(\text{subgraph $i$ disconnected} \mid \text{\# nodes} \geq rN_h/2) \\
	&\geq 1 - M (rN_h/2)^{1-\lambda},
	\end{align*}
	so the overall probability by the union bound is
	\begin{align*}
	&\Pr(\text{all subgraphs connected}) \\
	&\geq 1 - M (rN_h/2)^{1-\lambda} - M\exp \left( - \frac{r N_h}{8(1-r)} \right).
	\end{align*}
\end{proof}
While the theorem considered the honest subgraph, the full graph is $G(N,p)$, where $N = N_h/(1-\alpha)$ denotes the total number of nodes with $\alpha$ denoting the malicious fraction. Then, substituting $r = \frac{(L/k)\log L}{4L}$, each node should have 
\begin{align*}
pN &= \frac{2 \lambda \log (rN_h/2)}{rN_h} \frac{N_h}{1-\alpha} \\
&= 8 \lambda k \cdot \frac{1}{1-\alpha}\frac{ \log\left( \frac{N_h \log L}{8k} \right)}{\log L} \\
&= O\left( \left(\frac{1}{1-\alpha} \right)k \log N_h \right)
\end{align*}
neighbors.

\subsection{Collaborative Decoding Analysis}
In this section, we analyze the correctness of the collaborative decoding protocol described in Section~\ref{section:collabdecode}.
\begin{thm}\label{thm:collabdecodeanalysis}
    Suppose that each selective broadcast reaches the entire network. Also, suppose that each participant samples $c \log L$ symbols to decode, following the sampling procedure in Section~\ref{section:avail}, and $N_h \geq (L/c)(\log L + \lambda)$. Then, the collaborative decoding procedure in Section~\ref{section:collabdecode} satisfies the following:
    \begin{enumerate}
        \item If the code is constructed correctly and the data is available, i.e.\ no stopping set is hidden, then the nodes decode all of the data with probability $\geq 1- e^{-\lambda}$.
        \item If the data is unavailable, then all honest nodes reject the block with probability $\geq 1 - N_h (1-f)^c$, where $f$ is the size of the smallest stopping set.
        \item If the data is available but the code is constructed incorrectly, then all nodes receive a valid coding fraud proof with probability $\geq 1- e^{-\lambda}$.
    \end{enumerate}
\end{thm}
\begin{proof}
    Recall that in the collaborative decoding protocol, each participant chooses a subtree to decode. To decode this subtree, it keeps track of any parity equations connected to any symbols in the subtree, downloading any broadcasted symbols in these parity equations. Each node broadcasts any symbols it decoded, along with a Merkle proof.
    
    First, let's prove case (1). Suppose that each symbol in the tree is covered by at least one honest participant. Also, suppose that each selective broadcast reaches the entire network. Note that because broadcasted symbols must be accompanied by a Merkle proof, malicious nodes cannot broadcast fake symbols. Next, if there is a degree one parity equation at any point in the protocol, then by the coverage condition, there will exist at least one honest node in charge of decoding the last unknown symbol. By the protocol, this node will have downloaded each of the other symbols in the parity equation, so it will decode the final symbol and broadcast it. Therefore, decoding proceeds exactly as in the non-collaborative case, so it succeeds as long as no stopping set is hidden. By Theorem~\ref{thm:coverage}, the coverage condition holds with probability $\geq 1 - e^{-\lambda}$.
    
    To prove case (2), suppose that a stopping set is hidden in at least one of the layers of the coded Merkle tree. Each node samples $c$ symbols at random from each layer. Then, each node samples an unavailable share with probability $\geq 1 - (1-f)^c$. Therefore, by the union bound, every node samples an unavailable share with probability $\geq 1 - N_h (1-f)^c$.
    
    Finally, let's prove case (3). Suppose that the code is constructed incorrectly, meaning that there exists at least one parity equation that does not sum to zero. By availability, we can assume that eventually the entire data will be decoded. Let's assume both the coverage and selective broadcast conditions. We split into two cases:
    \begin{enumerate}
        \item Suppose that during decoding, all parity equations used to decode are correct, and invalid parities are not used. Then, decoding will succeed and all of the data will be revealed. After decoding, in step $2$ of $\texttt{collab\_decode\_layer}$ (Section~\ref{section:collabdecode}), each node checks its parity equations to ensure that they sum to zero. By coverage, at least one node will check the invalid parity equation and broadcast a coding fraud proof.
        \item Otherwise, suppose that in the process of decoding, at least one invalid parity equation is used. Let's consider the time at which the first invalid parity equation is decoded. Because this parity equation is the first invalid one, all symbols decoded until now are correct, meaning that they match the hash commitments. When the invalid parity equation is decoded, it must be degree one. In particular, if the equation sums to $K \neq 0$, we have the equation
        \begin{equation*}
            0 \neq K = s_d + \sum_{i=1}^{d-1} s_i,
        \end{equation*}
        where $s_d$ is the unknown symbol. When an honest node decodes the unknown symbol, the resulting symbol is
        \begin{equation*}
            \hat{s}_d = -\sum_{i=1}^{d-1} s_i \neq K - \sum_{i=1}^{d-1} s_i = s_d.
        \end{equation*}
        Next, following the protocol, the node will check the hash of $\hat{s}_d$. The true hash $\texttt{hash}(s_d)$ was decoded in the previous layer, and it is correct because the current parity equation is the first invalid one. Therefore, the hashes will not match, or $\texttt{hash}(\hat{s}_d) \neq \texttt{hash}(s_d)$, so the node will broadcast a coding fraud proof.
    \end{enumerate}
\end{proof}

\subsection{Security Analysis}
In this section, we show the correctness of the protocol as a whole. Let $N_h$ denote the number of honest nodes, $L$ denote the number of transactions per block, $1/k$ denote the fraction of each block that each participant validates, $\alpha$ denote the fraction of malicious nodes, and $\lambda$ denote the security parameter. 

Dividing the transactions into $k$ sections, let each honest participant validate one section at random following the protocol in Section~\ref{section:collabvalprotocol}. Also, let each honest participant randomly sample $(L/k)\log L$ symbols from the coded Merkle tree to decode, as described in Section~\ref{section:avail}. Then, we assume that the number of honest nodes is at least
\begin{align*}
    N_h\geq k(\log L + \lambda) 
\end{align*}
and the number of neighbors per node is at least
\begin{align*}
    \text{Neighbors per node} \geq 8 \lambda k \cdot \frac{1}{1-\alpha}\frac{ \log\left( \frac{N_h \log L}{8k} \right)}{\log L},
\end{align*}
where these expressions are derived from the conditions in Sections~\ref{section:coverage}, \ref{section:connectivity}, and \ref{section:collabdecodeanalysis}.
\begin{thm}
    Given the above setup, the protocol satisfies the following:
    \begin{enumerate}
        \item If the block is valid and available, then all honest participants will accept the block with probability $\geq 1 - e^{-\lambda} -  4L \left(\frac{N_h}{8k}\right)^{1-\lambda} - 4L \exp \left( - \frac{N_h}{8(4k-1)} \right)$.
        \item If the block is unavailable, then all honest participants will reject the block with probability $\geq 1 - N_h (1-f)^{L/k}$.
        \item If the block is invalid but available, then all honest participants will reject the block with probability $\geq 1 - e^{-\lambda} -  4L \left(\frac{N_h}{8k}\right)^{1-\lambda} - 4L \exp \left( - \frac{N_h}{8(4k-1)} \right)$.
    \end{enumerate}
    
\end{thm}
\begin{proof}
    Recall that by Theorem~\ref{thm:coverage}, all sections are validated by at least one honest participant with probability $\geq 1 - e^{-\lambda}$. Also, by Theorem~\ref{thm:connectivity}, and all selective broadcasts reach the entire network with probability $\geq 1 - 4L \left(\frac{N_h}{8k}\right)^{1-\lambda} - 4L \exp \left( - \frac{N_h}{8(4k-1)}\right)$.
    
    First, we will prove case (1). The only reasons for a participant to reject the block are (a) it receives a valid fraud proof, (b) it receives a valid coding fraud proof, or (c) it is unable to download or decode all of its sampled symbols. (a) and (b) cannot occur because the block is valid, so there do not exist valid fraud or coding fraud proofs. To show that (c) cannot occur, let's assume that each selective broadcast reaches the entire network. Then, by case (1) of Theorem~\ref{thm:collabdecodeanalysis}, the collaborative decoding procedure succeeds and all participants receive their sampled symbols with probability $\geq 1- e^{-\lambda}$. So by the union bound, the selective broadcast condition holds and collaborative decoding succeeds with probability $\geq 1 - e^{-\lambda} -  4L \left(\frac{N_h}{8k}\right)^{1-\lambda} - 4L \exp \left( - \frac{N_h}{8(4k-1)} \right)$.
    
    Next, case (2) follows directly from case (2) of Theorem~\ref{thm:collabdecodeanalysis}.
    
    Finally, we will prove case (3). Let's assume the selective broadcast condition. We will consider two cases.
    
    First, suppose that the block contains an invalid transaction but not coding fraud. Then, by case (1) of Theorem~\ref{section:collabdecodeanalysis}, the data is decoded successfully and each node receives its desired data. Then, assuming that all sections are validated by at least one honest participant, at least one participant will validate the invalid transaction and broadcast a fraud proof. Then, given that the network is connected, all participants receive the fraud proof, validate it, and reject the block. By the union bound, the selective broadcast and coverage conditions hold with probability $\geq 1 - e^{-\lambda} -  4L \left(\frac{N_h}{8k}\right)^{1-\lambda} - 4L \exp \left( - \frac{N_h}{8(4k-1)} \right)$.
    
    Otherwise, suppose that the block contains coding fraud. Then, by case (3) of Theorem~\ref{section:collabdecodeanalysis}, all nodes receive a valid coding fraud proof with probability $\geq 1 - e^{-\lambda}$, which is the probability that each symbol in the tree is covered by at least one honest participant. Then, the nodes will validate the coding fraud proof and reject the block. By the union bound, the selective broadcast and coverage conditions hold with probability $\geq 1 - e^{-\lambda} -  4L \left(\frac{N_h}{8k}\right)^{1-\lambda} - 4L \exp \left( - \frac{N_h}{8(4k-1)} \right)$.
\end{proof}

\end{document}